\documentclass[10pt,a4paper]{article}
\usepackage{latexsym,amsthm,amssymb,amscd,amsmath}
\newcounter{alphthm}
\theoremstyle{plain}
\newtheorem{theorem}{Theorem}[section]
\newtheorem{lemma}[theorem]{Lemma}

\newtheorem{cor}[theorem]{Corollary}
\theoremstyle{definition}

\newcommand{\be}{\begin{equation}}
\newcommand{\ee}{\end{equation}}
\newcommand{\ben}{\begin{enumerate}}
\newcommand{\een}{\end{enumerate}}

\begin{document}

\title{A K\"{a}hler Structure on Cartan Spaces}
\author{E. Peyghan and A. Tayebi}

\maketitle
\begin{abstract}
In this paper, we define a new metric on  Cartan manifolds  and
obtain a K\"{a}hler structure on their cotangent bundles. We prove that on a Cartan manifold $M$ of negative
constant flag curvature, $(T^{\ast}M_0,G,J)$ has  a K\"{a}hlerian structure. For Cartan manifolds of positive constant flag curvature, we show that the tube around the zero section has a K\"{a}hlerian structure. Finally by computing the Levi-Civita connection and components of curvature related to this metric, we show that
there is no non-Riemannian Cartan structure such that
$(T^{\ast}M_0, G, J)$ became a Einstein manifold or locally
symmetric manifold.
\end{abstract}


\section{Introduction}
The modern formulation of the notion of Cartan spaces is due to 
the R. Miron \cite{Mir1}, \cite{Mir2}, \cite{Mir3}.
 Based on the studies of E. Cartan, A. Kawaguchi \cite{Kaw}, R.
Miron \cite{MirHri}, \cite{Mir2}, \cite{Mir3}, S. Vacaru \cite{V1}, \cite{V2}, \cite{V3},  D. Hrimiuc and
H. Shimada \cite{Hri}, \cite{HriShi},
 P.L. Antonelli and  M. Anastasiei  \cite{Anas}, \cite{MirAn}, \cite{M2}, etc.,
 the geometry of Cartan spaces is today an important chapter of differential geometry.

Analytical Mechanics and some theories in Physics brought into
discussion regular Lagrangians and their geometry \cite{MirAn}.
 Regular Lagrangian which is 2-homogeneous in velocities is nothing but the
square of a fundamental Finsler function and its geometry is
Finsler geometry. This geometry was developed since 1918 by P.
Finsler, E. Cartan, L. Berwald, H. Akbar-Zadeh and many others, see \cite{Matsu}
and the most recent graduate texts  \cite{AZ}\cite{BCS}\cite{ShDiff}\cite{ShLec}. On the other hand, there is a strong linkage among  Finsler, Hamilton and Cartan geometries \cite{PT2}. For example,  Anastasiei and Vacaru provide a method of converting Lagrange and Finsler spaces and their Legendre transforms to Hamilton and Cartan spaces into almost K\"{a}hler structures on tangent and cotangent bundles \cite{AV}. But in Mechanics
and Physics there exists also regular Hamiltonians whose geometry
is also useful. This geometry is mainly due to Miron, and it is now systematically presented in the
monograph \cite{MirHri}. A manifold endowed with a regular
Hamiltonian which is 2-Homogeneous in momenta was called a Cartan space.

Let us denote the Hamiltonian structure on a manifold $M$ by $(M,
H(x,p))$. If the fundamental function $H(x,p)$ is 2-homogeneous
on the fibres of the cotangent bundle $(T^*M, M)$, then the
notion of Cartan space is obtained. It is remarkable that these
spaces appear as dual of the Finsler spaces, via Legendre
transformation. Using this duality several important results in
the Cartan spaces can be obtained: the canonical nonlinear
connection, the canonical metrical connection, etc. Therefore, the
theory of Cartan spaces has the same symmetry and beauty like
Finsler geometry. Moreover, it gives a geometrical framework for
the Hamiltonian theory of Mechanics or Physical fields.

Let $(M,K)$ be a Cartan space on a manifold $M$ and $\tau :=\frac{1}{2}K^2$. Let us  define the
symmetric $M$-tensor field $G_{ij}:=\frac{1}{\beta}g_{ij}+\frac{v(\tau )}{\alpha \beta }p_{i}p_{j}$ on $T^{\ast}M_0$,
where $v$ is  a real valued smooth function  defined on $[0,\infty)\subset {\mathbb R} $ and $\alpha $ and $\beta$ are real constants. Then we can define the following Riemannian metric on $T^{\ast}M_0$
\[
G=G_{ij}dx^{i}dx^{j}+G^{ij}\delta p_{i}\delta p_{j},
\]
where $G^{ij}$ is the inverse of $G_{ij}$. Then we define an almost complex structure  $J$  on  $T^{\ast}M_0$
by $J(\delta _{i})=G_{ik}\dot{\partial}^k$ and $J(\dot{\partial}^i)=-G^{ik}\delta _k$.

In this paper, we prove that  $(T^{\ast}M_0,G,J)$ is an almost K\"{a}hlerian manifold. Then we show that the almost complex structure $J$ on $T^{\ast}M_0$ is integrable if and only if $M$ has constant scalar curvature $c$ and  the function $v$ is given  by $v=-c\alpha\beta^2$. We conclude that on a Cartan manifold $M$ of negative
constant flag curvature, $(T^{\ast}M_0,G,J)$ has  a K\"{a}hlerian structure. For Cartan manifolds of positive constant flag curvature, we show that the tube around the zero section has a K\"{a}hlerian structure (see Theorem \ref{THM1}).

Then we find the Levi-Civita connection $\nabla$ of the metric
$G$. For the connection $\nabla$, we compute all of components
curvature. For a Cartan space $(M,K)$ of  constant curvature $c$,
we prove that in the following cases $(M,K)$ reduces to a
Riemannian space: $(i)$ for $c<0$, $(T^{\ast}M_0,G,J)$ became a
K\"{a}hler Einstein manifold,\ $(ii)$ for $c>0$,
$(T_{\beta}^{\ast}M_0,G,J)$ became a K\"{a}hler Einstein manifold,
where $T_{\beta}^{\ast}M_0$ is the tube in $T^{\ast}M_0$.
It results that, there is not any non-Riemannian Cartan structure
such that $(T^{\ast}M_0, G, J)$ became a Einstein manifold (see Theorem \ref{THM2}).

Finally, for a Cartan space $(M,K)$ of  constant curvature $c$, we prove
that in the following cases $(M,K)$  reduces to a Riemannian
space: $(i)$ for $c<0$, $(T^{\ast}M_0,G,J)$ is a locally symmetric
K\"{a}hler manifold,\ $(ii)$ for $c>0$,
$(T_{\beta}^{\ast}M_0,G,J)$ is a locally symmetric K\"{a}hler
manifold. Then we conclude that there is not any non-Riemannian
Cartan structure such that $(T^{\ast}M_0, G, J)$ became a locally
symmetric manifold (see Theorem \ref{THM3}).

\section{Preliminaries}
Let $M$ be an $n$-dimensional $C^{\infty}$ manifold and
$\pi^{\ast}:T^{\ast}M\rightarrow M$ its cotangent bundle. If
$(x^i)$ are local coordinates on $M$, then $(x^i,p_i)$ will be taken
as local coordinates on $T^{\ast}M$ with the momenta $(p_i)$
provided by $p=p_idx^i$ where $p\in T^{\ast}_xM$, $x=(x^i)$ and
$(dx^i)$ is the natural basis of $T^{\ast}_xM$. The indices $i, j,
k,\ldots$ will run from 1 to $n$ and the Einstein convention on
summation will be used.

Put $\partial_i:=\frac{\partial}{\partial
x^i}$ and $\dot{\partial}^i:=\frac{\partial}{\partial p_i}$. Then $(\partial_i, \dot{\partial}^i)$ be the
natural basis in $T_{(x,p)}T^{\ast}M$ and $(dx^i,dp_i)$ be the dual
basis of it. The kernel $V_{(x,p)}$ of the differential
$d\pi^{\ast}:T_{(x,p)}T^{\ast}M\rightarrow T_xM$ is called the
vertical subspace of $T_{(x,p)}T^{\ast}M$ and the mapping
$(x,p)\rightarrow V_{(x,p)}$ is a regular distribution on
$T^{\ast}M$ called the vertical distribution. This is
integrable with the leaves $T^{\ast}_xM$, $x\in M$ and is locally
spanned by $\dot{\partial}^i$. The vector field
$C^{\ast}=p_i\dot{\partial}^i$ is called the Liouville vector field
and $\omega=p_idx^i$ is called the Liouville 1-form on $T^{\ast}M$.
Then $d\omega$ is the canonical symplectic structure on $T^{\ast}M$.
For an easer handling of the geometrical objects on $T^{\ast}M$ it
is usual to consider a supplementary distribution to the vertical
distribution, $(x,p)\rightarrow N_{(x,p)}$, called the
horizontal distribution and to report all geometrical
objects on $T^{\ast}M$ to the decomposition
\begin{equation}
T_{(x,p)}T^{\ast}M=N_{(x,p)}\oplus V_{(x,p)}.\label{decom}
\end{equation}
The pieces produced by the decomposition (\ref{decom}) are called
\textit{d}-geometrical objects (\textit{d} is for distinguished)
since their local components behave like geometrical objects on $M$,
although they depend on $x=(x^i)$ and momenta $p=(p_i)$.

The horizontal distribution is taken as being locally spanned by the
local vector fields
\begin{equation}
\delta_i:=\partial_i+N_{ij}(x,p)\dot{\partial}^j.\label{decom1}
\end{equation}
The horizontal distribution is called also a nonlinear
connection on $T^{\ast}M$ and the functions $(N_{ij})$ are called
the local coefficients of this nonlinear connection. It is important
to note that any regular Hamiltonian on $T^{\ast}M$ determines a
nonlinear connection whose local coefficients verify
$N_{ij}=N_{ji}$. The basis $(\delta_i,\dot{\partial}^i)$ is adapted
to the decomposition (\ref{decom}). The dual of it is $(dx^i,\delta
P_i)$, for $\delta p_i=dp_i-N_{ji}dx^j$.

A Cartan structure on $M$ is a function
$K:T^{\ast}M\longrightarrow [0,\infty)$  which has the following properties: (i) $K$ is $C^{\infty}$ on $T^{\ast}M_0=T^{\ast}M-\{0\}$; (ii) $K(x,\lambda p)=\lambda K(x,p)$ for all $\lambda>0$ and (iii) the $n\times n$ matrix $(g^{ij})$, where $g^{ij}(x,p)=\frac{1}{2}\dot{\partial}^i\dot{\partial}^jK^2(x,p)$, is positive definite at all  points of $T^{\ast}M_0$. We notice that in fact $K(x,p)>0$, whenever $p\neq0$. The pair $(M,K)$ is called a Cartan space. Using this notations, let us define
\[
p^i=\frac{1}{2}\dot{\partial}^iK^2\ \  \textrm{and}\ \ C^{ijk}=-\frac{1}{4}\dot{\partial}^i\dot{\partial}^j\dot{\partial}^kK^2.
\]
The properties of $K$ imply that
\begin{eqnarray}
&&p^i=g^{ij}p_j,\ \ p_i=g_{ij}p^j,\ \
K^2=g^{ij}p_ip_j=p_ip^j,\nonumber\\
&&C^{ijk}p_k=C^{ikj}p_k=C^{kij}p_k=0.
\end{eqnarray}
One considers the formal Christoffel symbols by
\begin{equation}
\gamma^i_{jk}(x,p):=\frac{1}{2}g^{is}(\partial_kg_{js}+\partial_jg_{sk}-\partial_sg_{jk}),
\end{equation}
and the contractions
$\gamma^{\circ}_{jk}(x,p):=\gamma^i_{jk}(x,p)p_i$,
$\gamma^{\circ}_{j\circ}:=\gamma^i_{jk}p_ip^k$. Then the functions
\begin{equation}
N_{ij}(x,p)=\gamma^{\circ}_{ij}(x,p)-\frac{1}{2}\gamma^{\circ}_{h\circ}(x,p)\dot{\partial}^hg_{ij}(x,p),\label{decom2}
\end{equation}
define a nonlinear connection on $T^{\ast}M$. This nonlinear
connection was discovered by R. Miron \cite{Mir1}. Thus a
decomposition (\ref{decom}) holds. From now on we shall use only
the nonlinear connection given by (\ref{decom2}).

A linear connection $D$ on $T^{\ast}M$ is said to be an
$N$-linear connection if $D$ preserves by parallelism the
distribution $N$ and $V$, also we have $D\theta=0$, for
$\theta=\delta p_i\wedge dx^i$. One proves that an $N$-linear
connection can be represented in the adapted basis
$(\delta_i,\dot{\partial}^i)$ in the form
\begin{eqnarray}
D_{\delta_j}\delta_i\!\!\!\!&=&\!\!\!\!H^k_{ij}\delta_j,\ \ \
D_{\delta_j}\dot{\partial}^i=-H^i_{kj}\dot{\partial}^k,\\
D_{\dot{\partial}^j}\delta_i\!\!\!\!&=&\!\!\!\!V^{kj}_i\delta_k,\ \
\ D_{\dot{\partial}^j}\dot{\partial}^i=-V^{ij}_k\dot{\partial}^k,
\end{eqnarray}
where $V^{kj}_i$ is a \textit{d}-tensor field and $H^k_{ij}(x,p)$
behave like the coefficients of a linear connection on $M$. The
functions $H^k_{ij}$ and $V^{kj}_i$ define operators of
\textit{h}-covariant and \textit{v}-covariant derivatives in the
algebra of \textit{d}-tensor fields, denoted by $_{|k}$ and
$\mid^k$, respectively. For $g^{ij}$ these are given by
\begin{eqnarray}
{g^{ij}}_{|k}\!\!\!\!&=&\!\!\!\!\delta_kg^{ij}+g^{sj}H^i_{sk}+g^{is}H^j_{sk},\\
g^{ij}\!\!\mid^k\!\!\!\!&=&\!\!\!\!\dot{\partial}^kg^{ij}+g^{sj}V^{ik}_s+g^{is}V^{jk}_s.
\end{eqnarray}
An \textit{N}-linear connection given in the adapted basis
$(\delta_i,\dot{\partial}^j)$ as $D\Gamma(N)=(H^i_{jk},V^{ik}_j)$ is
called metrical if
\begin{equation}
{g^{ij}}_{|k}=0,\ \ \ \ g^{ij}\!\!\mid^k=0.
\end{equation}
One verifies that the \textit{N}-linear connection
$C\Gamma(N)=(H^i_{jk},C^{ik}_j)$ with
\begin{eqnarray}
H^i_{jk}\!\!\!\!&=&\!\!\!\!\frac{1}{2}g^{is}(\delta_jg_{sk}+\delta_kg_{js}-\delta_sg_{jk}),\\
C^{jk}_i\!\!\!\!&=&\!\!\!\!-\frac{1}{2}g_{is}(\dot{\partial}^jg^{sk}+\dot{\partial}^kg^{sj}-\dot{\partial}^sg^{jk})=g_{is}C^{sjk},
\end{eqnarray}
is metrical and its \textit{h}-torsion
$T^i_{jk}:=H^i_{jk}-H^i_{kj}=0$, \textit{v}-torsion
$S^{jk}_i:=C^{jk}_i-C^{kj}_i=0$ and deflection tensor
$\Delta_{ij}=N_{ij}-p_kH^k_{ij}=0$. Moreover, it is unique with
these properties. This is called the canonical metrical connection
of the Cartan space $(M,K)$. It has also the following properties:
\begin{eqnarray}
K^2_{|j}:\!\!\!\!&=&\!\!\!\!\delta_jK^2=0,\ \ K^2\!\!\mid^j=2p^j,\\
p_{i|j}\!\!\!\!&=&\!\!\!\! p^i_{|j}=0,\ \  p_i\!\!\mid^j=\delta^j_i,\ \ p^i\!\!\mid^j=g^{ij}.\label{4}\\
R_{kij}p^k\!\!\!\!&=&\!\!\!\!0, \ \  P^i_{jk}p^j=0,\ \ P^i_{jk}:=H^i_{jk}-{\dot{\partial}}^iN_{jk}.\label{5}\\
\delta_ig_{jk}\!\!\!\!&=&\!\!\!\!H^s_{ji}g_{sk}+H^s_{ki}g_{js}.\label{6}
\end{eqnarray}
\section{ K\"{a}hler Structures on $T^{\ast}M_0$}
Suppose that 
\begin{equation}
 \tau :=\frac{1}{2}K^2=\frac{1}{2}g^{ij}(x,p)p_{i}p_{j},
\end{equation}
we consider a real valued smooth function $v$ defined on
$[0,\infty)\subset {\mathbb R} $ and real constants $\alpha $ and
$\beta$. We define the following symmetric $M$-tensor field of type
(0,2) on $T^{\ast}M_0$ having the components
\begin{equation}
G_{ij}:=\frac{1}{\beta}g_{ij}+\frac{v(\tau )}{\alpha \beta }p_{i}
p_{j}.
\end{equation}
It follows easily that the matrix  $(G_{ij})$ is positive definite
if and only if  $\alpha ,\beta >0,\, \, \,\, \alpha +2\tau v>0.$ The
inverse of this matrix has the entries
\begin{equation}
G^{kl}=\beta g^{kl}-\frac{v\beta }{\alpha +2\tau v}p^{k} p^{l}.
\end{equation}
The components  $G^{kl} $ define symmetric  $M$-tensor field of
type (0,2) on  $T^{\ast}M_0$. It is easy to see that if the matrix
$(G_{ij} )$ is positive definite, then matrix $(G^{kl} )$ is
positive definite too. Using $(G_{ij})$ and $(G^{ij})$, the following Riemannian metric
 on $T^{\ast}M_0$ is defined
\begin{equation}
G=G_{ij}dx^{i}dx^{j}+G^{ij}\delta p_{i}\delta p_{j}.
\end{equation}

Now, we define an almost complex structure  $J$  on  $T^{\ast}M_0$
by
\begin{equation}
J(\delta _{i})=G_{ik}\dot{\partial}^k,\ \ \
J(\dot{\partial}^i)=-G^{ik}\delta _k.
\end{equation}
It is easy to check that $J^2=-I$.
\begin{theorem}
 $(T^{\ast}M_0,G,J)$ is an almost K\"{a}hlerian manifold.
\end{theorem}
\begin{proof} Since the matrix $(H^{kl} )$ is the inverse of the
matrix $(G_{ij} ),$ then we have
\[
G(J\delta_i,J\delta_j)=G_{ik}G_{jr}G(\dot{\partial}^k,\dot{\partial}^r)=G_{ik}G_{jr}G^{kr}=G_{ij}=G(\delta_i,\delta_j).
\]
The relations
\[
G(J\dot{\partial}^i,J\dot{\partial}^j)=G(\dot{\partial}^i,\dot{\partial}^j),\
\ \ \ G(J\delta _i,J\dot{\partial}^j)=G(\delta
_i,\dot{\partial}^j)=0,
\]
may be obtained in a similar way, thus
\[
G(JX,JY)=G(X,Y),\, \, \, \, \, \, \, \forall X,Y\in
\Gamma(T^{\ast}M_0).
\]
This means that $G$ is almost Hermitian with respect to $J$. The
fundamental 2-form associated by this almost K\"{a}hler structure is
$\Omega$, defined by
\[
\theta(X,Y):=G(X,JY),\, \, \, \, \forall X,Y\in \Gamma
(T^{\ast}M_0).
\]
Then we get
\[
\theta(\dot{\partial}^i,\delta_j)=G(\dot{\partial}^i,J\delta_j)=G(\dot{\partial}^i,G_{jk}\dot{\partial}^k)=G^{ik}G_{jk}=\delta^i_j,
\]
 and
\[
\theta(\delta_i,\delta_j)=\theta(\dot{\partial}^i,\dot{\partial}^j)=0.
\]
Hence, we have
\begin{equation}
\theta=\delta p_i\wedge dx^i,
\end{equation}
that is the canonical symplectic form of $T^{\ast}M$.
\end{proof}

Here, we study the integrability of the almost complex structure
defined by $J$ on $T^{\ast}M$. To do this, we need the following lemma.

\begin{lemma}\rm{(\cite{MirHri}\cite{PT})}
Let $(M,K)$ be a Cartan space.  Then we have the following:\\
$(1)\ \ [\delta_i,\delta_j]=R_{kij}{\dot{\partial}}^k$,\\
$(2)\ \ [\delta_i,{\dot{\partial}}^j]=-({\dot{\partial}}^jN_{ik}){\dot{\partial}}^k$,\\
$(3)\ \ [{\dot{\partial}}^i,{\dot{\partial}}^j]=0$,\\
where $R_{kij}=\delta_jN_{ki}-\delta_iN_{kj}$.
\end{lemma}

\smallskip

\begin{theorem}\label{3.3}
Let $(M, K)$ be a Cartan space. Then $J$ is a complex structure on
$T^{\ast}M_0$ if and only if $A_{kij}=0$ and
\begin{equation}
R_{kij}=\frac{v}{\alpha\beta^2}(g_{ik}p_j-g_{jk}p_i),\label{2}
\end{equation}
where $A_{kij}=\delta_i G_{jk}-\delta_j
G_{ik}+G_{ir}{\dot{\partial}}^rN_{jk}-G_{jr}{\dot{\partial}}^rN_{ik}$.
\end{theorem}\label{th1}
\begin{proof}
Using the definition of the Nijenhuis tensor field $N_J$ of $J$, that is,
\[
N_J(X,Y)=[JX,JY]-J[JX,Y]-J[X,JY]-[X,Y],\ \ \ \forall X, Y\in
\Gamma(T^{\ast}M)
\]
we get
\begin{equation}
N_J(\delta_i,\delta_j)=A_{hij}H^{hk}\delta_k+
(B_{kij}-R_{kij}){\dot{\partial}}^k,\label{Nij1}
\end{equation}
where
$B_{kij}=G_{ir}{\dot{\partial}}^rG_{jk}-G_{jr}{\dot{\partial}}^rG_{ik}$. Let $C_{jk}^r:=g_{jl}g_{sk}C^{rls}$, then we have
\[
{\dot{\partial}}^rg_{jk}=-g_{jl}g_{sk}{\dot{\partial}}^rg^{ls}=2g_{jl}g_{sk}C^{rls}=2C_{jk}^r.
\]
By the above equation, we obtain
\begin{equation}
G_{ir}{\dot{\partial}}^rG_{jk}=\frac{2}{\beta^2}C_{ijk}
+\frac{v}{\alpha\beta^2}(g_{ji}p_k+g_{ik}p_j)
+(\frac{v^{\prime}}{\alpha\beta^2}+\frac{2vv^{\prime}\tau}{\alpha\beta}+\frac{2v^2}{\alpha^2\beta^2})p_ip_jp_k.\label{Nij2}
\end{equation}
where $C_{ijk}=g_{ir}C_{jk}^r$. From (\ref{Nij2}) we get
\begin{equation}
B_{kij}=\frac{v}{\alpha\beta^2}(g_{ik}p_j-g_{jk}p_i).\label{Nij3}
\end{equation}
By a straightforward computation, it follows that
$N_J({\dot{\partial}}^i,{\dot{\partial}}^j)=0,
N_J({\dot{\partial}}^i,\delta_j)=0$, whenever
$N_j(\delta_i,\delta_j)=0$. Therefore, from relations (\ref{Nij1})
and (\ref{Nij3}), we conclude that the necessary and sufficient
conditions for  vanishing the Nijenhuis tensor field $N_J$ is obtained,  so that J
is a complex structure. Thus $A_{kij}=0$ and (\ref{2}) hold.
\end{proof}

\bigskip

In equation (\ref{2}),  we put $-\frac{v}{\alpha\beta^2}=c$, where $c$
is constant. Then we get
\begin{equation}
R_{kij}=c(g_{jk}p_i-g_{ik}p_j).\label{3}
\end{equation}
\begin{theorem}\label{3.4}
Let $(M,K)$ be a Cartan space of dimension $n\geq 3$. Then the
almost complex structure $J$ on $T^{\ast}M_0$ is integrable if and
only if (\ref{3}) is hold and the function $v$ is given by
\begin{equation}
v=-c\alpha\beta^2.\label{constant}
\end{equation}
\end{theorem}
\begin{proof}
From equation $p_{i|k}=0$ of relation (\ref{4}), we conclude that
$\delta_ip_k=N_{ik}$. Hence the tensor field $A_{kij}$ takes the
form $A_{kij}=\delta_i g_{jk}-\delta_j
g_{ik}+g_{ir}{\dot{\partial}}^rN_{jk}-g_{jr}{\dot{\partial}}^rN_{ik}$
and by part $(iii)$ of Proposition 2.3 in chapter 7 of \cite{MirHri}, it
vanishes for Cartan spaces. Now we let $v=-c\alpha\beta^2$, then
from equation $A_{kij}=0$ and Theorem \ref{3.3}, we conclude that $J$
is integrable if and only if (\ref{3}) is hold.
\end{proof}
A Cartan space $K^n$ is of constant scalar $c$ if
\begin{equation}
H_{hijk}p^ip^jX^hX^k=c(g_{hj}g_{ik}-g_{hk}g_{ij})p^ip^jX^hX^k,\label{7}
\end{equation}
for every $(x,p)\in T^{\ast}_0M$ and $X=(X^i)\in T_xM$. Here
$H_{hijk}$ is the (hh)h-curvature of the linear Cartan connection of
$K^n$.

We replace $H_{hijk}$ in (\ref{7}) with $g_{is}H^s_{hjk}$ and so it
reduce to
\begin{equation}
p_sH^s_{hjk}p^jX^hX^k=c(p_hp_k-K^2g_{hk})X^hX^k.\label{8}
\end{equation}
By part (ii) of Proposition 5.1 in chapter 7 of \cite{MirHri},
$p_sH^s_{hjk}=-R_{hjk}$, hence we get
\[
R_{hjk}p^jX^hX^k=c(K^2g_{hk}-p_hp_k)X^hX^k,
\]
or equivalently
\begin{equation}
R_{hjk}p^j=c(K^2g_{hk}-p_hp_k),\label{8}
\end{equation}
because $(X^h)$ and $X^k$ are arbitrary vector fields on $M$. It is
easy to check that (\ref{8}) follows from (\ref{3}). Hence we
conclude that if (\ref{3}) is hold, then Cartan space $K^n$ has the
constant scalar curvature $c$.
\begin{theorem}\label{THM1}
Let $(M,K)$ be a Cartan space with constant flag curvature $c$. Suppose that $v$ is given by (\ref{constant}). Then\\
$(i)$ for negative constant $c$, structure $(T^{\ast}M_0,G,J)$ is a K\"{a}hler manifold;\\
$(ii)$ for positive constant $c$, the tube around the zero section
in $T^{\ast}M$, defined by the condition $2\tau=K^2<\frac{1}{c\beta^2}$,  is a K\"{a}hler manifold.
\end{theorem}
\begin{proof}
The function  $uv$  must satisfies in the following condition
\begin{equation}
\alpha +2\tau u=\alpha (1+2(-c)\beta ^{2} \tau )>0,\, \, \, \, \,
\alpha ,\beta >0.
\end{equation}
By using the above relation and Theorem \ref{3.4}, we complete the proof.
\end{proof}

\bigskip

By attention to above theorem, the components of the K\"{a}hler
metric $G$ on $TM_0$ are given by
\begin{equation}
\left\{\begin{array}{l} {G_{ij} =\frac{1}{\beta } g_{ij} -c\beta
y_{i}y_{j} ,} \\ {H_{ij} =\beta g_{ij}+\frac{c\beta ^{3} }{1-2c\beta
^{2} \tau } y_{i}y_{j} .}
\end{array}\right.\label{Kahler}
\end{equation}
\section{A K\"{a}hler Einstein Structure on Cotangent Bundle}
In this section, we study the property of $(T^{\ast}M_0,G)$ to be
Einstein manifold.  We find the expression of the Levi-Civita connection $\nabla$ of the metric $G$ on $T^{\ast}M_0$. Then we get  the curvature tensor field of $\nabla$, and  by computing the corresponding traces we find the components of Ricci tensor field of   $\nabla$.
\subsection{The Levi-Civita connection of $G$}
Here we determine the Levi-Civita connection of the K\"{a}hler metric $G$.
Recall that for Cartan space with Cartan connection, the relation
$P^j_{ik}=H^j_{ik}-{\dot{\partial}}^jN_{ik}$ is hold, and so we
have
$[\delta_i,{\dot{\partial}}^j]=(P^j_{ik}-H^j_{ik}){\dot{\partial}}^k$.
Also the Levi-Civita connection $\nabla$ of the Riemannian
manifold $(T^{\ast}M_0,G)$ is obtained from the formula
\begin{eqnarray}
2G({\nabla }_{X} Y,Z)\!\!\!\!&=&\!\!\!\!X(G(Y,Z))+Y(G(X,Z))-Z(G(X,Y))\nonumber\\
\!\!\!\!&+&\!\!\!\!G([X,Y],Z)-G([X,Z],Y)-G([Y,Z],X),\, \, \, \, \,
\forall X,Y,Z\in \Gamma(T^{\ast}M_0),\label{Levi}
\end{eqnarray}
and is characterized by the conditions $\nabla G=0$ and $T=0$, where
$T$ is the torsion tensor of $\nabla$. By the above equation we have
\begin{eqnarray}
2G(\nabla_{{\dot{\partial}}^i}{\dot{\partial}}^j,{\dot{\partial}}^k)
\!\!\!\!&=&\!\!\!\!{\dot{\partial}}^iG({\dot{\partial}}^j,{\dot{\partial}}^k)
+{\dot{\partial}}^jG({\dot{\partial}}^i,{\dot{\partial}}^k)-{\dot{\partial}}^kG({\dot{\partial}}^i,{\dot{\partial}}^j)\nonumber\\
\!\!\!\!&=&\!\!\!\!{\dot{\partial}}^i(\beta
g^{jk}+\frac{c\beta^3}{1-2c\beta^2\tau}p^jp^k)+{\dot{\partial}}^j(\beta
g^{ik}+\frac{c\beta^3}{1-2c\beta^2\tau}p^ip^k)\nonumber\\
\!\!\!\!&-&\!\!\!\!{\dot{\partial}}^k(\beta
g^{ij}+\frac{c\beta^3}{1-2c\beta^2\tau}p^ip^j)\nonumber\\
\!\!\!\!&=&\!\!\!\!
2\beta(-C^{ijk}+\frac{c^2\beta^4}{(1-2c\beta^2\tau)^2}p^ip^jp^k+\frac{c\beta^2}{1-2c\beta^2\tau}g^{ij}p^k).\label{con7}
\end{eqnarray}
Also, we obtain
\begin{eqnarray*}
2G(\nabla_{{\dot{\partial}}^i}{\dot{\partial}}^j,\delta_k)
\!\!\!\!&=&\!\!\!\!-\delta_kG({\dot{\partial}}^i,{\dot{\partial}}^j)
-G([{\dot{\partial}}^i,\delta_k],{\dot{\partial}}^j)-G([{\dot{\partial}}^j,\delta_k],{\dot{\partial}}^i)\nonumber\\
\!\!\!\!&=&\!\!\!\!-\delta_k(\beta
g^{ij}+\frac{c\beta^3}{1-2c\beta^2\tau}p^ip^j)-{\dot{\partial}}^iN_{kr}(\beta
g^{rj}+\frac{c\beta^3}{1-2c\beta^2\tau}p^rp^j)\\
\!\!\!\!&-&\!\!\!\!{\dot{\partial}}^jN_{kr}(\beta
g^{ri}+\frac{c\beta^3}{1-2c\beta^2\tau}p^rp^i)\nonumber\\
\!\!\!\!&=&\!\!\!\!-\delta_k(\beta
g^{ij}+\frac{c\beta^3}{1-2c\beta^2\tau}p^ip^j)+(P^i_{kr}-H^i_{kr})(\beta
g^{rj}+\frac{c\beta^3}{1-2c\beta^2\tau}p^rp^j)\\
\!\!\!\!&+&\!\!\!\!(P^j_{kr}-H^j_{kr})(\beta
g^{ri}+\frac{c\beta^3}{1-2c\beta^2\tau}p^rp^i).
\end{eqnarray*}
Since ${g^{ij}}_{|k}=0$ and $\delta_kp^i=-p^rH^i_{rk}$ are hold for Cartan connection, then  by the above equation we get
\begin{eqnarray}
2G(\nabla_{{\dot{\partial}}^i}{\dot{\partial}}^j,\delta_k)\!\!\!\!&=&\!\!\!\!
\frac{c\beta^3}{1-2c\beta^2\tau}(H^i_{rk}p^j+H^j_{rk}p^i-H^i_{kr}p^j-H^j_{kr}p^i)p^r+
\beta(P^i_{kr}+P^j_{kr})g^{ri}-{g^{ij}}_{|k}\nonumber\\
 \!\!\!\!&=&\!\!\!\!\beta(P^i_{kr}g^{rj}+P^j_{kr}g^{ri}).\label{con8}
\end{eqnarray}
From (\ref{con7}) and (\ref{con8}), we obtain
\begin{equation}
\nabla_{{\dot{\partial}}^i}{\dot{\partial}}^j=\frac{\beta^2}{2}(P^i_{kr}g^{rj}+P^j_{kr}g^{ri})g^{hk}\delta_h+(-C^{ij}_h+c\beta
G^{ij}p_h){\dot{\partial}}^h.\label{con3}
\end{equation}
Since $G$ is a K\"{a}hler metric, then $(M,K)$ is of constant
curvature $c$, i.e., $R_{kij}=c(g_{jk}p_i-g_{ik}p_j)$. Hence, we
get
\begin{eqnarray}
2G(\nabla_{\delta_i}{\dot{\partial}}^j,\delta_k)\!\!\!\!&=&\!\!\!\!{\dot{\partial}}^jG(\delta_i,\delta_k)
-G([\delta_i,\delta_k],{\dot{\partial}}^j)\nonumber\\
\!\!\!\!&=&\!\!\!\!{\dot{\partial}}^j(\frac{1}{\beta}g_{ik}-c\beta p_ip_k)-R_{rik}(\beta g^{rj}+\frac{c\beta^3}{1-2c\beta^2\tau}p^rp^j)\nonumber\\
\!\!\!\!&=&\!\!\!\!(\frac{1}{\beta}{\dot{\partial}}^jg_{ik}-c\beta\delta^j_ip_k-c\beta\delta^j_kp_i)-(c\beta\delta^j_kp_i-c\beta\delta^j_ip_k)\nonumber\\
\!\!\!\!&=&\!\!\!\!2(\frac{1}{\beta}C^j_{ik}-c\beta\delta^j_kp_i),\label{con}
\end{eqnarray}
where $C^j_{ik}:=g_{il}g_{kh}C^{jlh}=g_{il}C^{jl}_k$. Since
${g^{jk}}_{|i}=0$ and $p^j_{|i}=p^k_{|i}=0$, then we have
\[
{G^{jk}}_{|i}=\beta{g^{jk}}_{|i}+\frac{c\beta^3}{1-2c\beta^2\tau}(p^j_{|i}p^k+p^k_{|i}p^j)=0.
\]
Hence, we obtain
\begin{eqnarray}
2G(\nabla_{\delta_i}{\dot{\partial}}^j,{\dot{\partial}}^k)\!\!\!\!&=&\!\!\!\!\delta_iG({\dot{\partial}}^j,{\dot{\partial}}^k)
+G([\delta_i,{\dot{\partial}}^j],{\dot{\partial}}^k)-G([\delta_i,{\dot{\partial}}^k],{\dot{\partial}}^j)\nonumber\\
\!\!\!\!&=&\!\!\!\!\delta_iG^{jk}+(P^j_{ir}-H^j_{ir})G^{rk}+(H^k_{ir}-P^k_{ir})G^{rj}\nonumber\\
\!\!\!\!&=&\!\!\!\!{G^{jk}}_{|i}-2G^{rk}H^j_{ir}+P^j_{ir}G^{rk}-P^k_{ir}G^{rj}\nonumber\\
\!\!\!\!&=&\!\!\!\!-2G^{rk}H^j_{ir}+P^j_{ir}G^{rk}-P^k_{ir}G^{rj}.\label{con1}
\end{eqnarray}
From (\ref{con}) and (\ref{con1}), we have
\begin{equation}
\nabla_{\delta_i}{\dot{\partial}}^j=(C^{jh}_i-c\beta
G^{jh}p_i)\delta_h+(\frac{1}{2}P^j_{ih}-\frac{1}{2}P^k_{ir}G^{rj}G_{kh}-H^j_{ih}){\dot{\partial}}^h.\label{con2}
\end{equation}
Also from equation
\[
\nabla_{{\dot{\partial}}^i}{\delta_j}-\nabla_{\delta_j}{\dot{\partial}}^i=[{\dot{\partial}}^i,\delta_j]
=(H^i_{jh}-P^i_{jh}){\dot{\partial}}^h,
\]
and (\ref{con2}) we obtain
\begin{equation}
\nabla_{{\dot{\partial}}^i}{\delta_j}=(C^{ih}_j-c\beta
G^{ih}p_j)\delta_h-\frac{1}{2}(P^i_{jh}+P^k_{jr}G^{ri}G_{kh}){\dot{\partial}}^h.\label{con4}
\end{equation}
Similarly, we have
\begin{eqnarray}
2G(\nabla_{\delta_i}\delta_j,{\dot{\partial}}^k)\!\!\!\!&=&\!\!\!\!{\dot{\partial}}^kG(\delta_i,\delta_j)
+G([\delta_i,\delta_j],{\dot{\partial}}^k)\nonumber\\
\!\!\!\!&=&\!\!\!\!-{\dot{\partial}}^kG_{ij}+R_{rij}G^{rk}\nonumber\\
\!\!\!\!&=&\!\!\!\!-{\dot{\partial}}^k(\frac{1}{\beta}g_{ij}-c\beta p_ip_j)+c\beta(g_{rj}p_i-g_{ri}p_j)( g^{rk}+\frac{c\beta^2}{1-2c\beta^2\tau}p^rp^k)\nonumber\\
\!\!\!\!&=&\!\!\!\!-\frac{2}{\beta}C^k_{ij}+2c\beta\delta^k_jp_i,\label{con5}
\end{eqnarray}
and
\begin{eqnarray}
2G(\nabla_{\delta_i}\delta_j,\delta_k)\!\!\!\!&=&\!\!\!\!\delta_iG(\delta_j,\delta_k)+\delta_jG(\delta_i,\delta_k)-\delta_kG(\delta_i,\delta_j)\nonumber\\
\!\!\!\!&=&\!\!\!\!\delta_iG_{jk}+\delta_jG_{ik}-\delta_kG_{ij}\nonumber\\
\!\!\!\!&=&\!\!\!\!G_{jr}H^r_{ik}+G_{rk}H^r_{ij}+G_{ir}H^r_{jk}+G_{rk}H^r_{ij}-G_{rj}H^r_{ik}-G_{ir}H^r_{jk}\nonumber\\
\!\!\!\!&=&\!\!\!\!2G_{rk}H^r_{ij}.\label{con6}
\end{eqnarray}
From (\ref{con5}) and (\ref{con6}), we conclude the following equation
\begin{equation}
\nabla_{\delta_i}\delta_j=H^h_{ij}\delta_h+(-\frac{1}{\beta^2}C_{ijh}+c\beta
G_{hj}p_i){\dot{\partial}}^h.\label{con9}
\end{equation}
\subsection{The Curvature Tensor}
Here, we are going to compute the curvature tensors of $\nabla$. Recall that the curvature
$K$ of $\nabla$ is obtained from the following relation
\begin{equation}
K(X,Y)Z={\nabla }_{X} {\nabla }_{Y} Z-{\nabla }_{Y}{\nabla }_{X}
Z-{\nabla }_{[X,Y]} Z,\ \ \  \forall
X,Y,Z\in\Gamma(TM).\label{Curvature1}
\end{equation}
Using (\ref{Curvature1}) we have
\begin{equation}
K({\dot{\partial}}^i,{\dot{\partial}}^j){\dot{\partial}}^k={\nabla
}_{{\dot{\partial}}^i} {\nabla
}_{{\dot{\partial}}^j}{\dot{\partial}}^k-{\nabla
}_{{\dot{\partial}}^j}{\nabla
}_{{\dot{\partial}}^i}{\dot{\partial}}^k .\label{Curvature2}
\end{equation}
By  (\ref{con3}) we get
\begin{eqnarray}
\nabla_{{\dot{\partial}}^i}\nabla_{{\dot{\partial}}^j}{\dot{\partial}}^k\!\!\!\!&=&\!\!\!\!
\frac{\beta^2}{2}{\dot{\partial}}^i((P^j_{lr}g^{rk}+P^k_{lr}g^{rj})g^{lh})\delta_h
+\frac{\beta^2}{2}(P^j_{lr}g^{rk}+P^k_{lr}g^{rj})g^{lh}\nabla_{{\dot{\partial}}^i}\delta_h\nonumber\\
\!\!\!\!&+&\!\!\!\!(-C^{jk,i}_h+c\beta G^{jk,i}p_h+c\beta
G^{jk}\delta^i_h){\dot{\partial}}^h+(-C^{jk}_h+c\beta
G^{jk}p_h)\nabla_{{\dot{\partial}}^i}{\dot{\partial}}^h.
\label{Curvature3}
\end{eqnarray}
Since $P^h_{rl}=H^h_{rl}-{\dot{\partial}}^hN_{rl}$,
$N_{rl}=p_hH^h_{rl}$ and $p_h{\dot{\partial}}^hN_{rl}=N_{rl}$, then
$p_hP^h_{rl}=0$. Also we have $p^lP^h_{rl}=p^rP^h_{rl}=0$ and
${\dot{\partial}}^kg^{ij}=-2C^{kij}$. Therefore the following
relation deduce from (\ref{Curvature3})
\begin{eqnarray}
\nabla_{{\dot{\partial}}^i}\nabla_{{\dot{\partial}}^j}{\dot{\partial}}^k\!\!\!\!&=&\!\!\!\!
\frac{\beta^2}{2}\Big[(P^{j,i}_{lr}g^{rk}+P^{k,i}_{lr}g^{rj}-2P^j_{lr}C^{irk}-2P^k_{lr}C^{irj}
-P^i_{rl}C^{jkr}-P^h_{rl}C^{jk}_hg^{ri})g^{sl}
\nonumber\\
\!\!\!\!&&\!\!\!\!-P^j_{lr}C^{lis}g^{rk}
-P^k_{lr}C^{lis}g^{rj}\Big]\delta_s
+\Big[c\beta p_s(G^{jk,i}+c\beta G^{jk}G^{ih}p_h-\beta C^{jki})\nonumber\\
\!\!\!\!&&\!\!\!\!
-\frac{\beta^2}{4}g^{lh}(P^j_{lr}P^i_{hs}g^{rk}+P^j_{lr}P^m_{hn}g^{rk}g^{ni}g_{ms}+P^k_{lr}P^i_{hs}g^{rj}+P^k_{lr}P^m_{hn}g^{rj}g^{ni}g_{ms})
\nonumber\\
\!\!\!\!&&\!\!\!\!+c\beta
G^{jk}\delta^i_s+C^{jk}_hC^{ih}_s-C^{jk,i}_s\Big]{\dot{\partial}}^s,\label{Curvature4}
\end{eqnarray}
where $P^{j,i}_{lr}={\dot{\partial}}^iP^j_{lr}$. Since
${\dot{\partial}}^i\tau=p^i$ and ${\dot{\partial}}^ip^j=g^{ij}$,
then we obtain
\begin{eqnarray}
c\beta p_s\{G^{jk,i}\!\!\!\!&&\!\!\!\!-G^{ik,j}+c\beta G^{jk}G^{ih}p_h-c\beta G^{ik}G^{jh}p_h\}=\nonumber\\
\!\!\!\!&&\!\!\!\!\frac{c^2\beta^4}{1-2c\beta^2\tau}(g^{ik}p^j-g^{jk}p^i+g^{jk}p^i-g^{ik}p^j)=0.\label{Curvature5}
\end{eqnarray}
With replacing  $i,j$ in (\ref{Curvature4}), setting this equations
in (\ref{Curvature2}) and attention (\ref{Curvature5}), we can obtain the following
\begin{eqnarray}
K({\dot{\partial}}^i,{\dot{\partial}}^j){\dot{\partial}}^k\!\!\!\!&=&\!\!\!\!
\frac{\beta^2}{2}\Big[(P^{j,i}_{lr}g^{rk}+P^{k,i}_{lr}g^{rj}-P^{i,j}_{lr}g^{rk}
-P^{k,j}_{lr}g^{ri}-P^j_{lr}C^{irk}+P^h_{rl}C^{ik}_hg^{rj}-P^h_{rl}C^{jk}_hg^{ri}
\nonumber\\
\!\!\!\!&&\!\!\!\!+P^i_{rl}C^{jkr})g^{sl}
+g^{rk}(P^i_{lr}C^{ljs}-P^j_{lr}C^{lis})-P^k_{lr}C^{lis}g^{rj}
+P^k_{lr}C^{ljs}g^{ri}\Big]\delta_s\nonumber\\
\!\!\!\!&&\!\!\!\!
+\Big[\frac{\beta^2}{4}g^{lh}\big(P^i_{lr}P^j_{hs}g^{rk}-P^j_{lr}P^i_{hs}g^{rk}
+P^i_{lr}P^m_{hn}g^{rk}g^{nj}g_{ms}-P^j_{lr}P^m_{hn}g^{rk}g^{ni}g_{ms}\nonumber\\
\!\!\!\!&&\!\!\!\!+P^k_{lr}P^j_{hs}g^{ri}-P^k_{lr}P^i_{hs}g^{rj}
+P^k_{lr}P^m_{hn}g^{ri}g^{nj}g_{ms}-P^k_{lr}P^m_{hn}g^{rj}g^{ni}g_{ms}\big)\nonumber\\
\!\!\!\!&&\!\!\!\!+C^{ik,j}_s-C^{jk,i}_s+c\beta G^{jk}\delta^i_s-c\beta G^{ik}\delta^j_s+C^{jk}_hC^{ih}_s
-C^{ik}_hC^{jh}_s\Big]{\dot{\partial}}^s.\label{Curvature6}
\end{eqnarray}
By  (\ref{Curvature1}) it follows that
\begin{equation}
K({\dot{\partial}}^i,\delta_j){\dot{\partial}}^k={\nabla
}_{{\dot{\partial}}^i} {\nabla
}_{\delta_j}{\dot{\partial}}^k-{\nabla }_{\delta_j}{\nabla
}_{{\dot{\partial}}^i}{\dot{\partial}}^k-\nabla_{[{\dot{\partial}}^i,\delta_j]}{\dot{\partial}}^k.\label{Curvature7}
\end{equation}
From (\ref{con2}), we have
\begin{eqnarray}
{\nabla }_{{\dot{\partial}}^i}{\nabla
}_{\delta_j}{\dot{\partial}}^k=(C^{kh,i}_j\!\!\!\!&-&\!\!\!\!c\beta
G^{kh,i}p_j-c\beta G^{kh}\delta^i_j)\delta_h+(C^{kh}_j-c\beta
G^{kh}p_j)\nabla_{{\dot{\partial}}^i}\delta_h\nonumber\\
\!\!\!\!&&\!\!\!\!+\frac{1}{2}\{P^{k,i}_{jh}-(P^l_{jr}g^{rk}g_{lh})^{,i}-2H^{k,i}_{jh}\}{\dot{\partial}}^h
\nonumber\\
\!\!\!\!&&\!\!\!\!+\frac{1}{2}\{P^k_{jh}-P^l_{jr}g^{rk}g_{lh}-2H^k_{jh}\}
\nabla_{{\dot{\partial}}^i}{\dot{\partial}}^h.\label{Curvature8}
\end{eqnarray}
Putting (\ref{con3}) and (\ref{con4}) in the above relation yields
\begin{eqnarray}
{\nabla }_{{\dot{\partial}}^i}{\nabla
}_{\delta_j}{\dot{\partial}}^k\!\!\!\!&=&\!\!\!\!\Big[C^{ks,i}_j-c\beta
G^{ks,i}p_j-c\beta G^{ks}\delta^i_j+C^{kh}_jC^{is}_h-c\beta
C^{is}_jG^{kh}p_j+c^2\beta ^2G^{kh}G^{is}p_jp_s\nonumber\\
\!\!\!\!&&\!\!\!\!-\frac{\beta
^2}{4}g_{ts}\big(P^l_{jr}P^i_{mt}g^{rk}g_{lh}g^{mh}+P^l_{jr}P^h_{mt}g^{rk}g_{lh}g^{mi}+2P^i_{mt}H^k_{jh}g^{mh}
+2P^h_{mt}H^k_{jh}g^{mi}\big)\nonumber\\
\!\!\!\!&&\!\!\!\!+\frac{\beta
^2}{4}g^{ls}(P^k_{jh}P^i_{rl}g^{rh}+P^k_{jh}P^h_{rl}g^{ri})\Big]\delta_s+
\Big[-\frac{1}{2}g_{ls}(C^{kh}_jP^l_{hr}g^{ri}+\frac{1}{2}P^{l,i}_{jr}g^{rk})\nonumber\\
\!\!\!\!&&\!\!\!\!+\frac{1}{2}c\beta^2g^{kh}p_j(P^i_{hs}+P^l_{hr}g^{ri}g_{ls})
+\frac{1}{2}c\beta^2p_s(P^k_{jh}g^{ih}-P^i_{jr}g^{rk}-2H^k_{jh}G^{ih}) \nonumber\\\!\!\!\!&&\!\!\!\!
+P^l_{jr}(C^{ikr}g_{ls}-2g^{rk}C^i_{ls})
-\frac{1}{2}(P^k_{jh}C^{ih}_s+C^{kh}_jP^i_{hs}-P^{k,i}_{js})\nonumber\\
\!\!\!\!&&\!\!\!\!  -H^{k,i}_{js}+H^k_{jh}C^{ih}_s\Big]{\dot{\partial}}^s.\label{Curvature9}
\end{eqnarray}
Similarly, we obtain
\begin{eqnarray}
{\nabla }_{\delta_j}{\nabla
}_{{\dot{\partial}}^i}{\dot{\partial}}^k\!\!\!\!&=&\!\!\!\!\Big[\frac{\beta
^2}{2}\delta _j(P^i_{rl}g^{lk}g^{sr})+\frac{\beta ^2}{2}\delta
_j(P^k_{rl}g^{li}g^{sr})+\frac{\beta
^2}{2}P^i_{rl}H^s_{jh}g^{lk}g^{hr}\nonumber\\\!\!\!\!&&\!\!\!\!
+\frac{\beta
^2}{2}P^k_{rl}H^s_{jh}g^{li}g^{hr}-C^{ik}_hC^{hs}_j+c\beta
^2C^{ik}_hg^{hs}p_j-c^2\beta ^2G^{ik}G^{hs}p_hp_j\Big]\delta
_s\nonumber\\\!\!\!\!&&\!\!\!\! +\Big[\frac{\beta
^2}{2}(P^i_{rl}g^{lk}g^{hr}+P^k_{rl}g^{li}g^{hr})(-\frac{1}{\beta^2}C_{jhs}+c\beta
G_{sh}p_j)+\delta_j(-C^{ik}_s+c\beta
G^{ik}p_s)\nonumber\\\!\!\!\!&&\!\!\!\!+(-C^{ik}_h+c\beta
G^{ik}p_h)(\frac{1}{2}P^h_{js}-\frac{1}{2}P^l_{jr}g^{rh}g_{ls}-H^h_{js})\Big]{\dot{\partial}}^s.\label{Curvature10}
\end{eqnarray}
Using (\ref{con3}), we get
\begin{eqnarray}
\nabla_{[{\dot{\partial}}^i,\delta_j]}{\dot{\partial}}^k\!\!\!\!&=&\!\!\!\!\frac{\beta^2}{2}[P^h_{rl}H^i_{jh}g^{kl}+P^k_{rl}H^i_{jh}g^{hl}
-P^h_{rl}P^i_{jh}g^{kl}-P^k_{rl}P^i_{jh}g^{hl}]g^{sr}\delta_s\nonumber\\
\!\!\!\!&&\!\!\!\!+[-H^i_{jh}C^{hk}_s+c\beta
H^i_{jh}G^{hk}p_s+P^i_{jh}C^{hk}_s-c\beta^2
P^i_{jh}g^{hk}p_s]{\dot{\partial}}^s.\label{Curvature11}
\end{eqnarray}
With a direct computation it follows that
\begin{eqnarray}
&&c^2\beta^2G^{kh}G^{is}p_jp_h+c^2\beta^2G^{ik}G^{hs}p_jp_h-c\beta
G^{ks,i}p_j-c\beta G^{kh}C^{is}_hp_j-c\beta
G^{hs}C^{ik}_hp_j\nonumber\\
&&=\frac{c^2\beta^4}{1-2c\beta^2\tau}g^{is}p^kp_j+\frac{c^2\beta^4}{1-2c\beta^2\tau}g^{ik}p^sp_j
+\frac{2c^3\beta^6}{(1-2c\beta^2\tau)^2}p^ip^kp^sp_j+2c\beta^2C^{iks}p_j\nonumber\\
&&\ \ \
-\frac{2c^3\beta^6}{(1-2c\beta^2\tau)^2}p^ip^kp^sp_j-\frac{c^2\beta^4}{1-2c\beta^2\tau}g^{ik}p^sp_j-\frac{c^2\beta^4}{1-2c\beta^2\tau}g^{is}p^kp_j-2c\beta^2
C^{kis}p_j\nonumber\\
&&=0.\label{Curvature12}
\end{eqnarray}
By using (\ref{Curvature7})-(\ref{Curvature12}) and attention to
relations $G^{ik}_{|j}=0, p_{s|j}=0$, one can obtains the following
\begin{eqnarray}
K({\dot{\partial}}^i,\delta_j){\dot{\partial}}^k\!\!\!\!&=&\!\!\!\!\Big[-c\beta
G^{ks}\delta^i_j+C^{ks,i}_j+C^{kh}_jC^{is}_h+C^{ik}_hC^{hs}_j+\frac{\beta^2}{4}P^k_{jh}P^i_{rl}g^{rh}g^{ls}\nonumber\\
\!\!\!\!&&\!\!\!\!+\frac{\beta^2}{4}(P^k_{jh}P^h_{rl}g^{ri}g^{ls}
-P^l_{jr}P^i_{lt}g^{rk}g^{ts}-2P^i_{mt}F^k_{jh}g^{mh}g^{ts}-2P^h_{mt}F^k_{jh}g^{mi}g^{ts})\nonumber\\
\!\!\!\!&&\!\!\!\!-\frac{\beta^2}{2}\big(\delta_j(P^i_{rl}g^{lk}g^{sr})
+\delta_j(P^k_{rl}g^{li}g^{sr})+P^i_{rl}F^s_{jh}g^{lk}g^{hr}+P^k_{rl}F^s_{jh}g^{li}g^{hr}\big)\nonumber\\
\!\!\!\!&&\!\!\!\!-\frac{\beta^2}{2}(P^h_{rl}F^i_{jh}g^{kl}g^{sr}+P^k_{rl}F^i_{jh}g^{hl}g^{sr}
-P^i_{jh}P^h_{rl}g^{kl}g^{sr}-P^i_{jh}P^k_{rl}g^{hl}g^{sr})\nonumber\\
\!\!\!\!&&\!\!\!\!-\frac{\beta^2}{4}P^l_{jr}P^h_{mt}g^{rk}g_{lh}g^{mi}g^{ts}\Big]
\delta_s+\Big[C^{ik}_{s|j}-\frac{1}{2}C^{kh}_jP^i_{hs}-\frac{1}{2}C^{kh}_jP^l_{hr}g^{ri}g_{ls}\nonumber\\
\!\!\!\!&&\!\!\!\!  +\frac{1}{2}c\beta^2P^l_{hr}g^{kh}g^{ri}g_{ls}p_j+\frac{1}{2}P^{k,i}_{js}-\frac{1}{2}P^{l,i}_{jr}g^{rk}g_{ls}
+\frac{1}{2}P^l_{jr}C^{irk}g_{ls}
-P^l_{jr}C^i_{ls}g^{rk}\nonumber\\
\!\!\!\!&&\!\!\!\!-H^{k,i}_{js}-P^i_{jh}C^{kh}_s+c\beta^2P^i_{jh}g^{hk}p_s+
\frac{1}{2}c\beta^2(P^k_{jh}g^{ih}p_s-P^i_{jr}g^{rk}p_s-P^k_{sl}g^{li}p_j)
\nonumber\\\!\!\!\!&&\!\!\!\!
-\frac{1}{2}(P^k_{jh}C^{ih}_s-P^l_{jr}C^i_{ls}g^{rk}-P^i_{rl}C^r_{js}g^{lk}-P^k_{rl}C^r_{js}g^{li}
-P^h_{js}C^{ik}_h)\Big]{\dot{\partial}}^s.\label{curvature13}
\end{eqnarray}
From (\ref{con2}) and (\ref{con9}), we have
\begin{eqnarray}
K(\delta_i,\delta_j){\delta_k}
\!\!\!\!&=&\!\!\!\!\Big[R^s_{kji}+\frac{1}{\beta^2}(C_{ikh}C^{hs}_j-C_{jkh}C^{hs}_i)
+c^2\beta^2(p_i\delta^s_j-p_j\delta^s_i)p_k\Big]\delta_s\nonumber\\\!\!\!\!&&\!\!\!\!
+\Big[\frac{1}{\beta^2}(C_{ikh}P^h_{jr}-C_{jkh}P^h_{ir})+\frac{1}{\beta^2}(C_{ikr|j}-C_{jkr|i})+cg_{hk}(P^h_{ir}p_j-P^h_{jr}p_i)\nonumber\\\!\!\!\!&&\!\!\!\!
+c^2\beta^2p_hp_k(P^h_{jr}p_i-P^h_{ir}p_j)\Big]{\dot{\partial}}^s,\label{curvature14}
\end{eqnarray}
and
\begin{eqnarray}
K(\delta_i,\delta_j){\dot{\partial}}^k\!\!\!\!&=&\!\!\!\!
\Big[\frac{1}{2}(P^k_{jh}C^{hs}_i-P^k_{ih}C^{hs}_j)
+\frac{\beta^2}{2}R_{hij}(P^h_{mr}g^{ms}g^{rk}+P^k_{mr}g^{ms}g^{rh})\nonumber\\
\!\!\!\!&&\!\!\!\!
+\frac{c\beta^2}{2}(P^k_{ih}g^{hs}p_j-P^k_{jh}g^{hs}p_i-P^s_{ir}g^{rk}p_j+P^s_{jr}g^{rk}p_i)\nonumber\\
\!\!\!\!&&\!\!\!\!
+\frac{1}{2}g^{rk}(P^m_{ir}C^{s}_{jm}-\frac{1}{2}P^m_{jr}C^{s}_{im})+C^{ks}_{j|i}-C^{ks}_{i|j}\Big]\delta_s\nonumber\\
\!\!\!\!&&\!\!\!\!
+\Big[\frac{1}{2}\Big(\delta_j(P^m_{ir}g^{rk}g_{ms})-\delta_i(P^m_{jr}g^{rk}g_{ms})\Big)
+\frac{1}{\beta^2}(C^{kh}_iC_{jhs}-C^{kh}_jC_{ihs})\nonumber\\
\!\!\!\!&&\!\!\!\!+R^k_{sij}+c\beta R_{hij}G^{hk}p_s-
\frac{1}{2}(P^k_{js|i}-P^k_{is|j})
+\frac{1}{4}(P^k_{jh}P^h_{is}-P^k_{ih}P^h_{js})\nonumber\\
\!\!\!\!&&\!\!\!\!
+\frac{1}{4}g_{ws}\ \big(P^k_{ih}P^w_{ju}g^{uh}-P^k_{jh}P^w_{iu}g^{uh}
+P^m_{jr}P^w_{im}g^{rk}-P^m_{ir}P^w_{jm}g^{rk}\big)\nonumber\\
\!\!\!\!&&\!\!\!\!
+\frac{1}{2}g^{rk}g_{mh}(P^m_{jr}H^h_{is}-P^m_{ir}H^h_{js})+\frac{1}{2}g^{uh}g_{ws}(P^w_{iu}H^k_{jh}-P^w_{ju}H^k_{ih})\nonumber\\
\!\!\!\!&&\!\!\!\!
+\frac{1}{4}(P^m_{ir}P^h_{js}g^{rk}g_{mh}-P^m_{jr}P^h_{is}g^{rk}g_{mh})\Big]{\dot{\partial}}^s.\label{curvature15}
\end{eqnarray}
Similarly, from (\ref{con3}) and (\ref{con4}), we get
\begin{eqnarray}
K({\dot{\partial}}^i,{\dot{\partial}}^j)\delta_k\!\!\!\!&=&\!\!\!\!\Big[-
\frac{\beta^2}{4}g^{us}(P^j_{kh}P^i_{uv}g^{vh}-P^i_{kh}P^j_{uv}g^{vh}+P^j_{kh}P^h_{uv}g^{vi}-P^i_{kh}P^h_{uv}g^{vj}\nonumber\\
\!\!\!\!&&\!\!\!\!+P^m_{kr}P^h_{uv}g^{vi}g^{jr}g_{mh}-P^m_{kr}P^h_{uv}g^{vj}g^{ir}g_{mh}
+P^m_{kr}P^i_{um}g^{jr}-P^m_{kr}P^j_{um}g^{ir})\nonumber\\
\!\!\!\!&&\!\!\!\!+C^{js,i}_k-C^{is,j}_k+C^{jh}_kC^{is}_h-C^{ih}_kC^{js}_h+c\beta
G^{is}\delta^j_k-c\beta
G^{js}\delta^i_k\Big]\delta_s\nonumber\\
\!\!\!\!&&\!\!\!\!
+\frac{1}{2}\Big[c\beta^2p_k(P^i_{hs}g^{jh}-P^j_{hs}g^{ih}+P^m_{hr}g^{jh}g^{ir}g_{ms}-P^m_{hr}g^{ih}g^{jr}g_{ms})
\nonumber\\
\!\!\!\!&&\!\!\!\!
+(P^m_{kr}g^{ir}g_{ms})^{,j}-(P^m_{kr}g^{jr}g_{ms})^{,i}+P^{i,j}_{ks}-P^{j,i}_{ks}+C^{ih}_kP^j_{hs}
\nonumber\\
\!\!\!\!&&\!\!\!\!
-C^{jh}_kP^i_{hs}+C^{ih}_kP^m_{hr}g^{jr}g_{ms}-C^{jh}_kP^m_{hr}g^{ir}g_{ms}+C^{ih}_sP^j_{kh}
\nonumber\\
\!\!\!\!&&\!\!\!\!-C^{jh}_sP^i_{kh}+C^i_{ms}P^m_{kr}g^{jr}-C^j_{ms}P^m_{kr}g^{ir}\Big]{\dot{\partial}}^s.\label{curvature16}
\end{eqnarray}
Finally, from (\ref{con3}), (\ref{con2}), (\ref{con4}) and
(\ref{con9}) we have
\begin{eqnarray}
K(\delta_i,{\dot{\partial}}^j)\delta_k\!\!\!\!&=&\!\!\!\!
\Big[\frac{1}{2}c\beta^2p_i(P^j_{kh}g^{hs}+P^s_{kr}g^{rj})-\frac{1}{2}(P^j_{kh}C^{hs}_i
+P^l_{kr}C^s_{li}g^{rj})-H^{s,j}_{ik}+C^{js}_{k|i}\nonumber\\
\!\!\!\!&&\!\!\!\!
+\frac{1}{2}g^{sl}(P^j_{lr}C^r_{ik}+P^h_{lr}C_{ikh}g^{rj}-c\beta^3P^j_{lr}g^{rh}G_{hk}p_i)
+c\beta^2P^j_{ih}g^{hs}p_k\nonumber\\
\!\!\!\!&&\!\!\!\!-\frac{1}{2}c\beta^2P^h_{lr}g^{rj}g^{sl}g_{hk}p_i-P^j_{ih}C^{hs}_k\Big]\delta_s
+\Big[cC^j_{is}p_k-\frac{1}{2}p^j_{ks|i}-\frac{1}{\beta^2}C^{jh}_kC_{ihs}
\nonumber\\
\!\!\!\!&&\!\!\!\!-\frac{1}{2}\delta_i(P^l_{kr}g^{rj}g_{ls})-\frac{1}{4}P^j_{kh}P^h_{is}+\frac{1}{4}P^j_{kh}P^l_{ir}g^{rh}g_{ls}
-\frac{1}{4}P^l_{kr}P^h_{is}g^{rj}g_{lh}+\frac{1}{\beta^2}C^{,j}_{iks}\nonumber\\
\!\!\!\!&&\!\!\!\!+\frac{1}{4}P^l_{kr}P^m_{il}g^{rj}g_{ms}+\frac{1}{2}P^l_{kr}H^h_{is}g^{rj}g_{lh}+\frac{1}{2}P^l_{hr}H^h_{ik}g^{rj}g_{ls}-c\beta
G_{sk}\delta^j_i-\frac{c}{\beta}C^j_{ik}p_s\nonumber\\
\!\!\!\!&&\!\!\!\!-\frac{1}{\beta^2}C_{ikh}C^{jh}_s-\frac{1}{2}P^l_{kr}H^j_{ih}g^{rh}g_{ls}+\frac{1}{2}P^j_{ih}P^h_{ks}
+\frac{1}{2}P^j_{ih}P^l_{kr}g^{rh}g_{ls}\Big]{\dot{\partial}}^s.\label{curvature17}
\end{eqnarray}
\begin{theorem}\label{THM2}
Let $(M,K)$ be a Cartan space of constant curvature $c$ and the
components of the metric $G$ are given by (\ref{Kahler}).
Then the following are hold if and only if $(M,K)$ is reduce to a Riemannian space.\\
$(i)$ for $c<0$, $(T^{\ast}M_0,G,J)$ is a K\"{a}hler Einstein manifold.\\
$(ii)$ for $c>0$, $(T_{\beta}^{\ast}M_0,G,J)$ is a K\"{a}hler
Einstein manifold, where $T_{\beta}^{\ast}M_0$ is the tube around the
zero section in $TM$ defined by $2\tau<\frac{1}{c\beta^2}$.
\end{theorem}
\begin{proof}
Let $(M,K)$ be a Cartan space. Then $C^{hi}_k$ and $P^h_{ik}$
are vanish and $H^i_{jk}$ is a function of $(x^h)$. Therefore
(\ref{curvature14}) reduce to
\begin{equation}
K(\delta_i,\delta_j)\delta_k=[R^s_{kji}+c^2\beta^2(p_i\delta^s_j-p_j\delta^s_i)p_k]\delta_s.\label{Ricc1}
\end{equation}
From Proposition 10.2 in chapter 4 of \cite{MirHri}, we have
$R_{kji}=-p_hR^h_{kji}$. Then we get
\begin{equation}
p_hR^h_{kji}=c(g_{kj}\delta^h_i-g_{ki}\delta^h_j)p_h.\label{Ricc2}
\end{equation}
Differentiating  (\ref{Ricc2}) with respect to $p_s$, taking
$p=0$, yields
\[
R^s_{kji}=c(g_{kj}\delta^s_i-g_{ki}\delta^s_j).
\]
By putting above equation in (\ref{Ricc1}), it follows that
\begin{eqnarray}
K(\delta_i,\delta_j)\delta_k\!\!\!\!&=&\!\!\!\!c\beta[(\frac{1}{\beta}g_{kj}-c\beta
p_kp_j)\delta^s_i-(\frac{1}{\beta}g_{ki}-c\beta\
p_kp_i)\delta^s_j]\delta_s\nonumber\\
\!\!\!\!&=&\!\!\!\!c\beta(G_{kj}\delta^s_i-G_{ki}\delta^s_j)\delta_s.\label{Ricc3}
\end{eqnarray}
Also from (\ref{curvature17}), we obtain
\begin{equation}
K({\dot{\partial}}^i,\delta_j)\delta_k=c\beta
G_{sk}\delta^i_j{\dot{\partial}}^s.\label{Ricc4}
\end{equation}
From (\ref{Ricc3}) and (\ref{Ricc4}), we have
\begin{eqnarray}
Ric(\delta_j,\delta_k)\!\!\!\!&=&\!\!\!\!G^{hi}G(K(\delta_i,\delta_j)\delta_k,\delta_h)+G_{hi}G(K({\dot{\partial}}^i,\delta_j)\delta_k,{\dot{\partial}}^h),\nonumber\\
\!\!\!\!&=&\!\!\!\!c\beta(G_{kj}\delta^s_i-G_{ki}\delta^s_j)G^{hi}G_{sh}+c\beta
G_{sk}\delta^i_jG_{hi}G^{sh}\nonumber\\
\!\!\!\!&=&\!\!\!\!cn\beta G_{jk}\nonumber\\
\!\!\!\!&=&\!\!\!\!cn\beta
G(\delta_j,\delta_k).\label{Ricc5}
\end{eqnarray}
Similarly from (\ref{Curvature6}) and (\ref{curvature13}),
respectively, we have
\begin{eqnarray}
K({\dot{\partial}}^i,{\dot{\partial}}^j){\dot{\partial}}^k=c\beta(G^{jk}\delta^i_s-G^{ik}\delta^j_s){\dot{\partial}}^s,\label{Ricc6}
\end{eqnarray}
and
\begin{eqnarray}
K(\delta_i,{\dot{\partial}}^j){\dot{\partial}}^k=c\beta
G^{ks}\delta^j_i\delta_s.\label{Ricc7}
\end{eqnarray}
Using (\ref{Ricc6}) and (\ref{Ricc7}), we conclude that
\begin{eqnarray}
Ric({\dot{\partial}}^j,{\dot{\partial}}^k)\!\!\!\!&=&\!\!\!\!G^{ih}G(K(\delta_i,{\dot{\partial}}^j){\dot{\partial}}^k,\delta_h)
+G_{ih}G(K({\dot{\partial}}^i,{\dot{\partial}}^j){\dot{\partial}}^k,{\dot{\partial}}^h)\nonumber\\
\!\!\!\!&=&\!\!\!\!c\beta G^{ks}\delta^j_iG^{ih}G_{hs}+c\beta(G^{jk}\delta^i_s-G^{ik}\delta^j_s)G_{hi}G^{hs}\nonumber\\
\nonumber\!\!\!\!&=&\!\!\!\!cn\beta G^{jk}\\
\!\!\!\!&=&\!\!\!\!cn\beta
G({\dot{\partial}}^j,{\dot{\partial}}^k).\label{Ricc8}
\end{eqnarray}
By (\ref{curvature13}) and (\ref{curvature15}),respectively,  we have
\begin{equation}
K(\delta_i,\delta_j){\dot{\partial}}^k=(R^k_{sij}+c\beta
R_{hij}G^{hk}p_s){\dot{\partial}}^s,\label{Ricc9}
\end{equation}
and
\begin{equation}
K({\dot{\partial}}^i,\delta_j){\dot{\partial}}^k=-c\beta
G^{ks}\delta^i_j\delta_s.\label{Ricc10}
\end{equation}
Using (\ref{Ricc9}) and (\ref{Ricc10}), we obtain
\begin{equation}
Ric(\delta_j,{\dot{\partial}}^k)=G^{ih}G(K(\delta_i,\delta_j){\dot{\partial}}^k,\delta_h)
+G_{ih}G(K({\dot{\partial}}^i,\delta_j){\dot{\partial}}^k,{\dot{\partial}}^h)=0.\label{Ricc11}
\end{equation}
From (\ref{curvature16}), we have
\begin{equation}
K({\dot{\partial}}^i,{\dot{\partial}}^j)\delta_k=c\beta(G^{is}\delta^j_k-G^{js}\delta^i_k)\delta_s.\label{Ricc12}
\end{equation}
By attention to (\ref{Ricc4}) and (\ref{Ricc12}), one can obtains
\begin{equation}
Ric({\dot{\partial}}^j,\delta_k)=G^{ih}G(K(\delta_i,{\dot{\partial}}^j)\delta_k,\delta_h)
+G_{ih}G(K({\dot{\partial}}^i,{\dot{\partial}}^j)\delta_k,{\dot{\partial}}^h)=0.\label{Ricc13}
\end{equation}
From (\ref{Ricc5}), (\ref{Ricc8}), (\ref{Ricc11}) and
(\ref{Ricc13}), we deduce that
\[
Ric(X,Y)=cn\beta G(X,Y),\ \ \ \forall X, Y\in\chi(T^{\ast}M).
\]
This means that $(T^{\ast}M,G)$ is a Einstein manifold.

Conversely, suppose that the conditions  $(i), (ii)$ are hold. Then there exists a real constant $\lambda$ such that $Ric(X,Y)=\lambda G(X,Y)$. We consider the following cases:\\\\
\textbf{Case (1)}. If $\lambda=0$ (i.e., $(T^{\ast}M,G)$ is Ricci
flat), then we have
$Ric({\dot{\partial}}^j,{\dot{\partial}}^k)=0$. By using
(\ref{Curvature6}),(\ref{curvature13}) and definition of Ricci
tensor, we obtain
\begin{eqnarray}
0=p_kRic({\dot{\partial}}^j,{\dot{\partial}}^k)=cn\beta
p_kG^{jk}- p_kC^{jk,s}_s\!\!\!\!&+&\!\!\!\!\frac{\beta^2}{2}\delta_s(P^j_{rl}g^{lk}g^{sr})p_k
+\frac{\beta^2}{2}\delta_s(P^k_{rl}g^{lj}g^{sr})p_k\nonumber\\
\!\!\!\!&+&\!\!\!\!
\frac{\beta^2}{2}g^{ts}p_k(P^j_{mt}H^k_{sh}g^{mh}+P^h_{mt}H^k_{sh}g^{mj}).\label{Ricc14}
\end{eqnarray}
Since $P^j_{rl}g^{lk}g^{sr}p_k=P^j_{rl}p^lg^{sr}=0$ and
$\delta_sp_k=H^h_{sk}p_h$, then we have
\begin{eqnarray*}
\delta_s(P^j_{rl}g^{lk}g^{sr})p_k\!\!\!\!&=&\!\!\!\!-(\delta_sp_k)P^j_{rl}g^{lk}g^{sr}\\
\!\!\!\!&=&\!\!\!\!-P^j_{rl}g^{lk}g^{sr}H^h_{sk}p_h.
\end{eqnarray*}
Replacing $h$ and $k$, changing $r$ to $t$, and $l$ to $m$ in the above equation yields
\begin{equation}
\delta_s(P^j_{rl}g^{lk}g^{sr})p_k=-P^j_{tm}g^{mh}g^{st}H^k_{sh}p_k.\label{Ricc15}
\end{equation}
Similarly, we have
\begin{equation}
\delta_s(P^k_{rl}g^{lj}g^{sr})p_k=-P^h_{tm}g^{mj}g^{st}H^k_{sh}p_k.\label{Ricc16}
\end{equation}
With direct computation, we get
\begin{eqnarray}
\nonumber cn\beta p_kG^{jk}\!\!\!\!&=&\!\!\!\!cn\beta p_k(\beta
g^{jk}+\frac{c\beta^3}{1-2c\beta^2\tau}p^jp^k)\\
\!\!\!\!&=&\!\!\!\!\frac{cn\beta^2}{1-2c\beta^2\tau}p^j,\label{Ricc17}
\end{eqnarray}
and
\begin{equation}
p_kC^{jk,s}_s=-p^{,s}_kC^{jk}_s=-\delta^s_kC^{jk}_s=-C^{js}_s=-I^j.\label{Ricc18}
\end{equation}
By using (\ref{Ricc14})-(\ref{Ricc18}), we obtain
\begin{equation}
\frac{cn\beta^2}{1-2c\beta^2\tau}p^j+I^j=0.\label{Ricc19}
\end{equation}
Since $p_jI^j=0$, then by contracting (\ref{Ricc19}) with $p_j$, we have
\begin{equation}
\frac{2cn\beta^2\tau}{1-2c\beta^2\tau}=0.\label{Ricc20}
\end{equation}
From the above equation,  we conclude that $\beta=0$ and this is a contradiction.\\\\
\textbf{Case (2)}. If $\lambda\neq0$, then we have
\[
p_kRic({\dot{\partial}}^j,{\dot{\partial}}^k)=\lambda G^{jk}p_k.
\]
Therefore by using (\ref{Curvature6}), (\ref{curvature13}),
(\ref{Ricc15})-(\ref{Ricc18}), we obtain
\begin{equation}
I^j=(\lambda-cn\beta)\frac{\beta}{1-2c\beta^2\tau}p^j.\label{Ricc21}
\end{equation}
By contracting (\ref{Ricc21}) with $p_j$, we have
\begin{equation}
(\lambda-cn\beta)\frac{2\beta\tau}{1-2c\beta^2\tau}=0,\label{Ricc22}
\end{equation}
i.e., $\lambda=cn\beta$. Then from (\ref{Ricc21}), we result that
$I^j=0$, i.e., $(M,K)$ is reduce to a Riemannian space.
\end{proof}
\begin{cor}
There is not exist any non-Riemannian Cartan structure such that
$(T^{\ast}M_0, G, J)$ became a Einstein manifold.
\end{cor}

\bigskip

\begin{theorem}\label{THM3}
Let $(M,K)$ be a Cartan space of constant curvature $c$ and the
components of the metric $G$ are given by (\ref{Kahler}).
 Then the following are hold if and only if $(M,K)$ is reduce to a Riemannian space.\\
$(i)$ for $c<0$, $(T^{\ast}M_0,G,J)$ is a locally symmetric K\"{a}hler manifold.\\
$(ii)$ for $c>0$, $(T_{\beta}^{\ast}M_0,G,J)$ is a locally
symmetric K\"{a}hler manifold, where $T_{\beta}^{\ast}M_{\circ}$ is
the tube around the zero section in $T^{\ast}M$, defined by the
condition $2\tau<\frac{1}{c\beta^2}$.
\end{theorem}
\begin{proof}
Let $(M,K)$ be  a Cartan space. Using (\ref{Curvature6}),
(\ref{curvature13})-(\ref{curvature17}), we have computed the
covariant derivatives of curvature tensor field $K$ in the local
adapted frame $(\delta_i,\dot{\partial}^i)$ with respect to the
connection $\nabla$ and  obtained in the twelve cases the result
is zero.

Conversely, let $(i)$ and $(ii)$ are hold. Thus we get $\nabla K=0$. For simplify, we let $(M,K)$ be a Berwald-Cartan
space. Put
\[
(\nabla
K)({\dot{\partial}}^u,{\dot{\partial}}^i,\delta_j,{\dot{\partial}}^k)=M^{uiks}_j\delta_s+N^{uik}_{sj}{\dot{\partial}}^s.
\]
Since $(\nabla
K)({\dot{\partial}}^u,{\dot{\partial}}^i,\delta_j,{\dot{\partial}}^k)=0$,
then we have $M^{uiks}_j=N^{uik}_{sj}=0$. Thus we have
$p^jp_iM^{uiks}_j=0$. By a straightforward calculation,  we obtain
\begin{eqnarray}
p^jp_iM^{uiks}_j=p^jp_iC^{ks,i,u}_j\!\!\!\!&-&\!\!\!\!2c\beta\tau(G^{ks,u}+G^{kh}C^{us}_h+G^{uh}C^{ks}_h+G^{hs}C^{uk}_h)\nonumber\\
\!\!\!\!&+&\!\!\!\! 2c^2\beta^2\tau(G^{us}G^{kh}+G^{uk}G^{hs})p_h.\label{symmetric1}
\end{eqnarray}
By a direct computation, we have
\begin{equation}
G^{ks,u}=-2\beta
C^{ksu}+\frac{2c^2\beta^5}{(1-2c\beta^2\tau)^2}p^kp^up^s+\frac{c\beta^3}{1-2c\beta^2\tau}g^{ku}p^s+\frac{c\beta^3}{1-2c\beta^2\tau}g^{su}p^k.\label{symmetric2}
\end{equation}
Also, we obtain
\begin{equation}
G^{us}G^{kh}p_h=\frac{\beta}{1-2c\beta^2\tau}p^sG^{ku}=\frac{\beta^2}{1-2c\beta^2\tau}g^{ku}p^s+\frac{c\beta^4}{(1-2c\beta^2\tau)^2}p^up^kp^s,\label{symmetric3}
\end{equation}
and
\begin{equation}
G^{uk}G^{hs}p_h=\frac{\beta}{1-2c\beta^2\tau}p^kG^{su}=\frac{\beta^2}{1-2c\beta^2\tau}g^{su}p^k+\frac{c\beta^4}{(1-2c\beta^2\tau)^2}p^up^sp^k.\label{symmetric4}
\end{equation}
Putting (\ref{symmetric2}), (\ref{symmetric3}) and
(\ref{symmetric4}) in (\ref{symmetric1}), one can yields
\begin{equation}
p^jp_iM^{uiks}_j=p^jp_iC^{ks,i,u}_j-2c\beta\tau(-2\beta
C^{kus}+G^{kh}C^{us}_h+G^{uh}C^{ks}_h+G^{hs}C^{uk}_h).\label{symmetric5}
\end{equation}
Since $p_iC^{ks,i}_j=-C^{ks}_j$, then we obtain
\[
p^jp_iC^{ks,i,u}_j=2g^{uj}C^{ks}_j=2C^{kus}.
\]
Also, since
$C^{us}_hp^h=0$, then we have
\[
G^{kh}C^{us}_h=G^{uh}C^{ks}_h=G^{hs}C^{uk}_h=\beta C^{kus}.
\]
Therefore (\ref{symmetric5}) reduces to following
\begin{equation}
p^jp_iM^{uiks}_j=2(1-c\beta^2\tau)C^{kus}.\label{symmetric6}
\end{equation}
Since $p^jp_iM^{uiks}_j=0$, then by (\ref{symmetric6}), we have
$C^{kus}=0$, i.e., $K$  is a Riemannian metric.
\end{proof}
\begin{cor}
There is not exist any non-Riemannian Cartan structure such that
$(T^{\ast}M_0, G, J)$ became a locally symmetric manifold.
\end{cor}

\noindent
Esmail Peyghan\\
Faculty  of Science, Department of Mathematics\\
Arak University\\
Arak,  Iran\\
Email: epeyghan@gmail.com

\bigskip

\noindent
Akbar Tayebi\\
Faculty  of Science, Department of Mathematics\\
Qom University \\
Qom. Iran\\
Email: akbar.tayebi@gmail.com

\end{document}